\documentclass[11pt]{article}
\usepackage{enumerate}
\usepackage{amsmath,epsfig, amsthm, amssymb, enumerate}
\usepackage{fullpage}
\usepackage{graphicx}
\usepackage{subfigure}
\usepackage{amsmath, amssymb, epsfig}
\usepackage{bm}
\usepackage{mathrsfs}
\usepackage{color}

\graphicspath{{./Images/}}

\newtheorem{theorem}{Theorem}[]
\newtheorem{lemma}[theorem]{Lemma}
\newtheorem{prop}[theorem]{Proposition}
\newtheorem{corollary}[theorem]{Corollary}

\newtheorem{remark}[theorem]{Remark}

\def\x{{\mathbf x}}

\newcommand{\R}{\mathbb{R}}

\newcommand{\E}{\mathbb{E}}
\renewcommand{\P}{\mathbb{P}}

\newcommand{\vct}[1]{\bm{#1}}

\usepackage{algorithm}
\usepackage{algorithmic}

\title{A unified framework for linear dimensionality reduction in L1}

\author{Felix Krahmer and Rachel Ward}

\date{\today}

\begin{document}
\maketitle

\abstract{
For a family of interpolation norms $\| \cdot \|_{1,2,s}$ on $\mathbb{R}^n$, we provide a distribution over random matrices $\Phi_s \in \mathbb{R}^{m \times n}$ parametrized by sparsity level $s$ such that for a fixed set $X$ of $K$ points in $\mathbb{R}^n$, if $m \geq C s \log(K)$ then with high probability, $\frac{1}{2} \| \vct{x} \|_{1,2,s} \leq \| \Phi_s (\vct{x}) \|_1 \leq 2 \| \vct{x} \|_{1,2,s}$ for all $\vct{x} \in X$.  Several existing results in the literature roughly reduce to special cases of this result at different values of $s$: For $s=n$, $\| \vct{x} \|_{1,2,n} \equiv \| \vct{x} \|_{1}$ and we recover that dimension reducing linear maps can preserve the $\ell_1$-norm up to a distortion proportional to the dimension reduction factor, which is known to be the best possible such result. For $s=1$, $\| \vct{x} \|_{1,2,1} \equiv \| \vct{x} \|_{2}$, and we recover an $\ell_2 / \ell_1$ variant of the Johnson-Lindenstrauss Lemma for Gaussian random matrices.  Finally, if $\vct{x}$ is $s$-
sparse, then $\| \vct{x} \|_{1,2,s} = \| \vct{x} \|_1$ and we recover that 
$s$-sparse vectors in $\ell_1^n$ embed into $\ell_1^{\mathcal{O}(s \log(n))}$ via sparse random matrix constructions.
}


\section{Introduction} 
The theory for linear dimensionality reduction in Euclidean space has been the subject of much research in recent years.
The celebrated \emph{Johnson-Lindenstrauss (JL) Lemma} says that a small set of points in high-dimensional Euclidean space can be linearly embedded into a space of much lower dimension in such a way that Euclidean distances between the points are nearly preserved \cite{JL}.  More specifically, given a finite set $X \subset \mathbb{R}^n$ of size $| X | = K$, there exists a  linear map $\Phi: \mathbb{R}^n \rightarrow \mathbb{R}^m$ with $m = 9 \varepsilon^{-2} \log{(K)}$ such that $(1-\varepsilon) \| \vct{x} - \vct{y} \|_2 \leq \| \Phi(\vct{x-y}) \|_2 \leq (1+\varepsilon) \| \vct{x} - \vct{y} \|_2$ for all $\vct{x}, \vct{y} \in X$.  In the language of geometric embeddings, this says that $K$-point subsets of $\ell_2^n$ can be linearly embedded into $\ell_2^m$ with  $m = \mathcal{O}(\varepsilon^{-2} \log{(K)})$ and distortion $1 + \mathcal{O}(\varepsilon)$.  
 Remarkably, for a fixed finite set $X \subset \mathbb{R}^n$, taking $\Phi: \mathbb{R}^n \rightarrow \mathbb{R}^m$ as a random matrix whose entries are independent and identically-distributed mean-zero Gaussian random variables will achieve such an embedding with high probability \cite{dasgupta2003elementary}. Because such a probabilistic embedding is easy to construct and is \emph{oblivious} to the content of $X$, random projections have become an efficient pre-processing step for a wide range of algorithms in numerical linear algebra, compressive sensing, manifold learning, and theoretical computer science \cite{ind-01,sar,hmt, liwomaroty07,randomprojections,badadewa08, rw09, caneel10, kw12}.

It is natural to ask about embedding results for more general $\ell_p$ norms, $1 \leq p \leq \infty$, e.g., a result of the form $(1- \varepsilon)\| \vct{x} \|_{\ell_p^n} \leq \| \Phi \vct{x} \|_{\ell_p^m} \leq (1 + \varepsilon) \| \vct{x} \|_{\ell_p^n}$ for $m \ll n$.   Unfortunately, the strong results realized by Gaussian random matrices is specific to the case $p=2$.  In fact,  an embedding result of the Johnson-Lindenstrauss type is impossible for $p \neq 2$ using \emph{any} linear embedding: as shown by Charikar and Sahai  \cite{charikar2002dimension} for $\ell_1$ and generalized by Lee, Mendel, and Naor in \cite{lmn05} to $\ell_p$ for $1 \leq p \leq \infty$,  there are arbitrarily large $K$-point subsets $X$ of $\ell_p$  such that any linear mapping $T: X \rightarrow \ell_p^m$ incurs distortion at least $D = \Omega\left( (\frac{K}{m})^{\left| 1/p - 1/2 \right| } \right)$.   In particular, for $\ell_1$, dimensionality reduction is possible in general only if we allow for large distortion; that is, for 
an arbitrary finite set $X$ of $|X| = K$ points in $\mathbb{R}^n$, if one wishes for a map $T: \ell_1^n \rightarrow \ell_1^m$ such that
$\frac{1}{\sqrt{D}} \| \vct{x} - \vct{y} \|_1 \leq \| \Phi(\vct{x-y}) \|_1 \leq \sqrt{D} \| \vct{x} - \vct{y} \|_1$ holds for all $\vct{x,y} \in X$, then necessarily $m \geq C D^{-2} K$.  This bound is tight; see \cite{newman2010finite, schechtman2011dimension}.  The $\ell_1$ norm is of particular interest for several reasons,  one of which being that in high dimensions, the $\ell_1$ norm is more meaningful than the $\ell_2$ norm (and much more meaningful than the $\ell_p$ norm for $p$ large) for inferring neighborliness in large data sets  \cite{bgrs99, hak2000}. 

Still, the lower bounds for linear dimension reduction in $\ell_1$ represent a  \emph{worst-case} bound over arbitrary sets of points.  Restricting attention to structured subsets of points in $\ell_1^n$, much stronger statements can be made.  Of particular interest is the subset of \emph{sparse} vectors, where we recall that $\vct{x} \in \mathbb{R}^n$ is $s$-sparse if it has non-zero coordinates in at most $s$ dimensions.   In \cite{charikar2002dimension}, it was shown that an arbitrary set of $K$ $s$-sparse vectors can be linearly embedded into $\ell_1^{m}$ with distortion $1 + \varepsilon$ once $m \geq C\varepsilon^{-2} s^2 \log{K}$.  A uniform result over $s$-sparse vectors was subsequently shown by Berinde, Gilbert, Indyk, Karloff, and Strauss in  \cite{berinde2008combining}, which we state as a proposition.
\begin{prop}[From \cite{berinde2008combining}]
\label{expandgraph}
Fix $n, s \in \mathbb{N}$ and $\varepsilon \in (0,1)$, and fix $m \in \mathbb{N}$ satisfying $m \geq C\varepsilon^{-2} s \log{n}$.  There exist matrices $\Phi \in \mathbb{R}^{m \times n}$ such that 
\begin{equation}
\label{1RIP}
(1 - \varepsilon) \| \vct{x} \|_{\ell_1^n}  \leq \| \Phi \vct{x} \|_{\ell_1^m} \leq  (1+\varepsilon)\| \vct{x} \|_{\ell_1^n} \quad \quad \forall \vct{x} \in \mathbb{R}^n: \hspace{2mm} \| \vct{x} \|_0 \leq s.
\end{equation}
Such a matrix $\Phi$ is said to have the \emph{1-restricted isometry property} of order $s$ and level $\varepsilon$, or 1-RIP for short.   
\end{prop}
Explicit constructions of such matrices are \emph{binary} and \emph{sparse}; specifically, a sparse binary random matrix having $d = C \varepsilon^{-1} \log(n)$ ones per column and having $m \geq C' \varepsilon^{-1} s d = C'' \varepsilon^{-2} s \log(n)$ rows will have the 1-RIP with high probability  \cite{berinde2008combining, fr13}.   The 1-RIP property is essentially equivalent to the combinatorial notion of \emph{expansion}
 of the sparse bipartite graph underlying the measurement matrix.
 

\subsection{Contribution of this work} The aim of this work is to initiate a unified framework for linear dimension reduction in $\ell_1$.  We provide a general theorem which roughly interpolates between several existing distinct results.

More specifically, we consider a family of rearrangement-invariant block $\ell_1/\ell_2$ norms $\|   \cdot \|_{1,2,s}$ on $\mathbb{R}^n$ parametrized by block size $s \in [n]$, as follows.  For a given $\vct{x} \in \mathbb{R}^n$, consider a partition of its support into disjoint subsets $S_1, S_2, \dots $ so that $S_1$ indexes the largest $s$ elements of $\vct{x}$ in magnitude, $S_2$ indexes the next $s$ largest elements, and so on, and $S_{\lceil{n/s \rceil}}$ may contain between zero and $s-1$ elements.   Here on out, we will refer to the (not necessarily unique) decomposition $\vct{x} = \left( \vct{x}_{S_1}, \vct{x}_{S_2}, \dots  \right)$ as the \emph{$s$-block decreasing rearrangement} of $\vct{x}$.  The norm of interest is
\begin{equation}
\label{thenorm}
\| \vct{x} \|_{1,2,s} := \sqrt{ \sum_{\ell =1}^{\lceil{n/s \rceil}} \| \vct{x}_{S_{\ell}} \|_1^2 }
\end{equation}
for completeness, we verify in the appendix that this indeed defines a norm.   We call this an interpolation norm as on the one extreme, $s=1$ and $\| \cdot \|_{1,2,s} = \| \cdot \|_2$; at the other extreme, $s = n$ and $\| \cdot \|_{1,2,s} = \| \cdot \|_1.$

Together with this interpolation norm, we consider random matrices $\Phi_s \in \mathbb{R}^{m \times n}$ of the form
\begin{equation}
\label{thebed}
\Phi_s = A_s \circ G
\end{equation}
where $\circ$ denotes the Hadamard (entrywise) product, and
\begin{enumerate}
\item $G = (g_{j,k}) \in \mathbb{R}^{m \times n}$ is a random matrix populated with independent and identically distributed standard Gaussian entries, and 
\item $A_s = (a_{j,k}) \in \mathbb{R}^{m \times n}$ is a random matrix populated with zeros and ones, having exactly $d = \frac{m}{s}$ ones per column, the locations of which are chosen uniformly from $[m]$ without replacement. 
\end{enumerate}
Our main result is as follows. 
\begin{theorem}
\label{main}
Fix $m, n \in \mathbb{N}$ and parameter $\varepsilon > 0$.  Fix $s \in [n]$ such that $m/s \in \mathbb{N}$, and consider the random matrix $\Psi_s =  \sqrt{ \frac{2}{\pi} } \frac{s}{m} \Phi_s$.  For any fixed $\vct{x} \in \mathbb{R}^n$, it holds with probability exceeding $1 - 4n \exp( - \varepsilon^2 m / 8 ) -  2 \exp\left( -\varepsilon^2 \beta_0^2(m/(8s)) \left(\| \vct{x} \|_{1,2,s}/ \| \vct{x} \|_2 \right)^2 \right)$ that 
 \begin{equation}
\label{coarse:embed}
 (.63 - \varepsilon) \| \vct{x}  \|_{1,2,s}  \leq  \| \Psi_s \vct{x} \|_{1}  \leq  (1.63 +\varepsilon) \| \vct{x}  \|_{1,2,s}. 
 \end{equation}
\end{theorem}

\begin{remark}
\emph{
The block norm $\| \cdot \|_{1,2,s}$ has appeared previously in \cite{gg84, litvak2005smallest}, and is closely related to the so-called \emph{${\cal K}$-interpolation norm} $K(\vct{x},t) = \sum_{j=1}^{t^2} x_j^{*} + t(\sum_{j > t^2} (x_j^{*})^2)^{1/2}$ (where $\vct{x}^{*}$ denotes the decreasing rearrangement of $\vct{x}$), which is well-known in the theory of interpolation of Banach spaces \cite{bh88, h70}, and was used in the related context of upper and lower bounds for Rademacher sums in \cite{montgomery90}. As shown in Proposition \ref{exp:upper}, these two norms are equivalent up to a factor of $1.63$:
$$
\| \vct{x} \|_{1,2,s}  \leq {\cal K}(\x,\sqrt{s}) \leq 1.63 \| \vct{x} \|_{1,2,s}.
$$
}
\end{remark}

\bigskip

\noindent Theorem \ref{main} roughly interpolates between three existing results in the literature.
\begin{enumerate}
\item
For block-size $s=1$,  the block norm $\| \cdot \|_{1,2,s}$ coincides with the Euclidean norm $\| \cdot \|_{2}$, and $\Psi_s$ reduces to a properly-normalized  i.i.d. Gaussian random matrix.  Up to distortion $(1.63 + \varepsilon)/(.63 - \varepsilon) \approx 2.6 + \mathcal{O}(\varepsilon)$ instead of distortion $1 + \mathcal{O}(\varepsilon)$, Theorem \ref{main} recovers a well-established $\ell_2/\ell_1$ Johnson-Lindenstrauss concentration result for Gaussian random matrices, see for instance Lemma 5.3 of \cite{pv13}.

\item If $\vct{x} \in \mathbb{R}^n$ is $s$-sparse, then $\| \vct{x} \|_{1,2,s} = \| \vct{x} \|_1$.  In this case, Theorem \ref{main} with parameter $s$ recovers the 1-RIP embedding result from Proposition \ref{expandgraph} for sparse vectors, albeit only for the particular $s$-sparse vector $\vct{x}$ and not all $s$-sparse vectors.  Insufficient concentration of the Gaussian matrix $G$ prevents us from passing to a uniform sparse embedding result using $\Phi_s = A_s \circ G$ at number of measurements $m = \mathcal{O}(s \log(n))$; however, as shown in Corollary \ref{ballsinbins_cor}, we do recover the result of Proposition \ref{expandgraph} (up to distortion $(1.63 + \varepsilon)/(.63 - \varepsilon) \approx 2.6 + \mathcal{O}(\varepsilon)$ instead of distortion $1 + \mathcal{O}(\varepsilon)$) if we use for embedding the binary matrix $A_s$ alone, rather than the composite matrix $\Phi_s = A_s \circ G$.  

\item For block-size $s= n/D^2$ in Theorem \ref{main} with fixed constant distortion $D \geq 1$, the bock norm $\|\cdot \|_{1,2,s}$ is equivalent to the $\| \cdot \|_{1}$ norm up to a multiplicative factor of $D$.  In this case, Theorem \ref{main} produces an explicit embedding which realizes, up to a factor of $\log{(n)}$,  the best-possible dependence of $m = \Omega(n/D^2)$ dimensions necessary for embedding an arbitrary $n$-point set in $\ell_1$ into $\ell_1^m$ with constant distortion $D$.  

\end{enumerate} 
While we expect that the lower and upper distortion bounds $.63 - \varepsilon$ and $1.63+\varepsilon$ can be somewhat reduced, it is not possible to bring them down to $1 - \varepsilon$ and $1 + \varepsilon$, respectively. In Proposition~\ref{prop:counterex}, we provide two classes of examples and show that between them there is a constant distortion factor even in the asymptotic limit.

\begin{remark}
\emph{
It is noteworthy that, in contrast to the hardness results for linear dimension reduction in $\ell_1$, it \emph{is} possible to approximate the $\ell_1$ norm of a vector $\vct{x} \in \mathbb{R}^n$ up to distortion $1+ \varepsilon$ from a linear projection $\Phi \x$ into $m = \mathcal{O}(\varepsilon^{-2})$ dimensions if one considers other functions of $\Phi \x$ besides the $\ell_1$ norm.  The first result of this kind, provided in \cite{indyk06stable}, shows that taking $\Phi$ to have i.i.d. Cauchy-distributed entries (or i.i.d. from a $p$-stable distribution more generally for approximating the $\ell_p$ norm, $0 < p < 2$), $\| \vct{x} \|_1 = (1 \pm \varepsilon) \text{median}_k | (\Phi x)_k |$ with high probability and with $m = \mathcal{O}(\varepsilon^{-2})$.  Several subsequent works \cite{li08, knw10, nw10, tz12} provide estimators other than the median and random embedding matrices $\Phi$ other than Cauchy random matrices which are more optimized for practical implementations. 
}
\end{remark}

\subsection{Outline of the proof of Theorem \ref{main}} The proof of Theorem \ref{main} naturally splits in two parts.  First, we treat $A_s$ as fixed, so that $G$ is the only source of randomness in $\Phi = A_s \circ G$.  In  subsection \ref{sec:conditional} we use concentration estimates for sums of half-normal random variables to show that $\| \Phi \vct{x} \|_1$ has subgaussian concentration about its mean with variance $\sigma^2 = \frac{s}{m}\| \vct{x} \|_2^2$.  We consider $G$ as a Gaussian random matrix because $\mathbb{E}_{G} \| \Phi \vct{x} \|_1$ can be calculated explicitly: 
\begin{equation}
\label{exactformula}
\mathbb{E}_{G} ( \| \Phi \vct{x} \|_1) = \sqrt{\frac{2}{\pi}} \sum_{j=1}^m \left( \sum_{k=1}^n a_{j,k} x_k^2 \right)^{1/2}.
\end{equation}
In subsection \ref{sec:upper} we show, using only that $A_s = (a_{j,k}) \in \{0,1\}^{m \times n}$ has $m/s$ ones per column,
  \begin{equation}
  \label{upper_bound}
\sum_{j=1}^m \left( \sum_{k=1}^n a_{j,k} x_k^2 \right)^{1/2}  \leq   \frac{1.63m}{s}  \| \vct{x} \|_{1,2,s}.
  \end{equation}
In the second part of the proof, in subsection \ref{sec:balls}, we treat $A_s$ as random, and show that for fixed  $\vct{x} \in \mathbb{R}^n$, a lower bound is obtained with high probability with respect to the realization of $A_s$:
  $$
\sum_{j=1}^m \left( \sum_{k=1}^n a_{j,k} x_k^2 \right)^{1/2}  \geq (.63 - \varepsilon)\frac{m}{s} \| \vct{x} \|_{1,2,s} .
  $$
This lower bound follows rather directly from an estimate we show for each $s$-sparse component $\vct{x}_{S_{\ell}}$ of the block decreasing rearrangement of $\vct{x}$:
\begin{equation}
\label{weak:expand}
\sum_{j=1}^m \left( \sum_{k \in S_{\ell}} a_{j,k} x_k^2 \right)^{1/2} \geq (.63 - \varepsilon)\frac{m}{s}  \| \vct{x}_{S_{\ell}} \|_1.
\end{equation}
In words, \eqref{weak:expand} amounts to showing that by drawing $A \in \{0,1\}^{m \times n}$ having $m/s$ per column at random, we do not deviate much from the ideal situation where $A$, restricted to the columns indexed by $S_{\ell}$, contains \emph{exactly} one non-zero entry in each row, whence $\sum_{j=1}^m \left( \sum_{k \in S_{\ell}} a_{j,k} x_k^2 \right)^{1/2} = \frac{m}{s}\| \vct{x}_{S_{\ell}} \|_1$.   To verify that \eqref{weak:expand} is satisfied with sufficiently high probability, we use a ``balls into bins" analysis and borrow techniques from \cite{berinde2008combining} and \cite{fr13} used to show that a similar matrix construction satisfies the 1-RIP in Proposition \ref{expandgraph}.  We note that our binary matrix construction differs from those constructions in that we use exactly $m = sd$ measurements ($d$ is the number of ones per column), allowing us to apply the bound \eqref{upper_bound} with minimal constant 1.63.
  
\bigskip

\noindent As shown in Section \ref{proof:thm}, these ingredients can be combined to prove Theorem \ref{main}.  
  
\begin{remark}
\emph{
Without much additional effort, the Gaussian random matrix $G$ in Theorem \ref{main} can be replaced by a Bernoulli random matrix $B$, that is, a matrix whose entries are independent random variables which are $\pm 1$ with equal probability.   Indeed, embeddings of the form $\Phi = A_s \circ B$ are more appealing from a practical point of view as their entries are contained in $\{0,1,-1\}$, and are thus  more easily implemented and stored. Using a Bernoulli random matrix, the exact formula \eqref{exactformula} for the conditional expectation will no longer hold, but by the Khintchine inequality (using the optimal constants provided in \cite{kbest}), we still have 
$$
\frac{1}{\sqrt{2}} \sum_{j=1}^m \left( \sum_{k=1}^n a_{j,k} x_k^2 \right)^{1/2} \leq \mathbb{E}_{B} ( \| (A_s \circ B) \vct{x} \|_1) \leq  \sum_{j=1}^m \left( \sum_{k=1}^n a_{j,k} x_k^2 \right)^{1/2}.
$$
As a result, the analog of Theorem \ref{main} using a Bernoulli matrix will have slightly worse constants, but otherwise remains unchanged.  
}
\end{remark}

\subsection{Notation}  Throughout the paper, $C, c, C_1, \dots$ denote absolute constants whose values may change from line to line.  For integer $n$, we write $[n] = \{1,2,\dots, n\}$.  Vectors are written in bold italics, e.g. $\vct{x}$, and their coordinates written in plain text, e.g.  the $i$-th component of $\vct{x}$ is $x_i$.  For a subset $S \subset [n]$, $\vct{x}_S$ is the vector $\vct{x}$ restricted to the elements indexed by $S$, and may be treated as a dense vector in $\mathbb{R}^s$ or a sparse vector in $\mathbb{R}^n$ depending on the context.   Similarly, we denote by $A_S$ the matrix $A$ restricted to the columns indexed by $S$.

We denote by $\ell_p^n$ the space $\mathbb{R}^n$ with the $\ell_p$ norm   $\| \vct{x} \|_p = \left( \sum_{i=1}^n |x_i|^p \right)^{1/p}, 1 \leq p < \infty,$ $\| \vct{x} \|_{\infty}  =\max_i | x_i |$.  When the ambient dimension is not important, we simply write $\ell_p$.  The number of non-zero coordinates of a vector $\vct{x}$ is denoted by $\| \vct{x} \|_0 = | \text{supp}(\vct{x}) |$.

We write $f(u) = \mathcal{O}(g(u))$ if and only if there exists a positive real number $M$ and a real number $u_0$ such that $|f(u)| \leq M |g(u)|$ for all $u > u_0$.  We write $f(u)=\Omega(g(u))$ if and only if $g(u) = \mathcal{O}(f(u))$.  For a function $f$ of random variables $\vct{X} = (X_1, X_2)$, we write $\mathbb{E}_{X_1} f(\vct{X}) := \mathbb{E} \left[ f(\vct{X}) | X_2 \right]$ for the conditional expectation, which is itself a function of the random $X_2$.  

Finally, recall that an embedding $f: X \rightarrow Y$ of a metric space $(X,d)$ into a metric space $(Y, d')$ is said to have distortion $D \geq 1$ if there are constants $A, B \geq 1$ satisfying $AB \leq D$ such that for all $\vct{x,y} \in X$, 
$$A^{-1} d(\vct{x,y}) \leq d'(f(\vct{x}), f(\vct{y})) \leq B d(\vct{x,y}).$$

\section{Proof ingredients}
The remainder of the paper is devoted to proving Theorem \ref{main}.  We provide proof ingredients in this section, and we put the ingredients together in Section \ref{proof:thm}. 

{\subsection{Concentration lemmas}\label{sec:conditional}}

\noindent 
In this section we treat the binary matrix $A= A_s$ as fixed, and study the concentration of $\| \Phi \vct{x} \|_1 = \| (A \circ G) \vct{x} \|_1$ around its conditional expectation $\mathbb{E}_{G} \| (A \circ G)  \vct{x} \|_1$.   The first lemma is straightforward.

\begin{lemma}
\label{lemma:mean}
Fix $\vct{x} \in \mathbb{R}^n$ and $A = (a_{j,k} ) \in \{0,1\}^{m \times n}$.  Let $G = (g_{j,k}) \in \mathbb{R}^{m \times n}$ have i.i.d. Gaussian entries, and consider the random matrix $A \circ G$. The expectation of $ \| (A \circ G) \vct{x} \|_1$ is given by
$$
\mathbb{E}_{G} \| (A \circ G) \vct{x} \|_1 = \sqrt{\frac{2}{\pi}} \sum_{j=1}^m \left( \sum_{k=1}^n a_{j,k} x_{k}^2  \right)^{1/2}.
$$
\end{lemma}

\begin{proof}
The $j$th coordinate of $(A \circ G) \vct{x} \in \mathbb{R}^m$ can be written as $Y_j = \sum_{k=1}^n a_{j,k} g_{j,k} x_k;$
this is a mean-zero Gaussian random variable with variance $\sigma_j^2 = \sum_{k=1}^n a_{j,k} x_k^2$.  It follows that the random variable $|Y_j|$ is a half-normal random variable; it has mean $\mathbb{E} ( |Y_j | ) =  \sqrt{\frac{2}{\pi}} \sigma_j = \sqrt{\frac{2}{\pi}} \left( \sum_{k=1}^n a_{j,k} x_k^2 \right)^{1/2}$.  The lemma follows by noting that $\mathbb{E}_{G} \| (A \circ G) \vct{x} \|_1= \sum_{j=1}^m \mathbb{E}(| Y_j |)$. 
\end{proof}

We now show that if $A \circ G$ is a random matrix as above, and if $A$ has exactly $d$ ones per column, then $\| (A \circ G) \vct{x} \|_1$ exhibits subgaussian concentrates around its expectation $\mathbb{E}_{G} \| (A \circ G)  \vct{x} \|_1$.  To do this, we will need the following lemma, whose proof uses standard arguments and can be found in the appendix. 

\begin{lemma}
\label{concentrate}
Fix $d, m \in \mathbb{N}$ and $\alpha > 0$.  Suppose that $Y_i, \hspace{.5mm} i = 1,2, \dots, m,$ are independent mean-zero Gaussian random variables with variances $\sigma^2_i$ satisfying $\sum_{i=1}^m \sigma_i^2 = d \alpha^2$. Then the random variable $Z = \frac{1}{d} \sum_i ( |Y_j|  - \mathbb{E} |Y_j | )$ satisfies 
$$
\emph{Prob} \left[ |Z| \geq \lambda \right]  \leq 2 \exp{\left(-\frac{\lambda^2 d}{2\alpha^2} \right)}, \quad \quad \forall \lambda \geq 0.
$$
\end{lemma}
\noindent  Lemma \ref{concentrate} implies the following result about the concentration of $\| (A \circ G) \vct{x} \|_1$.
\begin{prop}
\label{main:concentrate}
Fix $\vct{x} \in \mathbb{R}^n$, and fix $A  = (a_{j,k}) \in \{0,1\}^{m \times n}$ populated with zeros and ones and having $d$ ones per  column.  Suppose that $G \in \mathbb{R}^{m \times n}$ consists of i.i.d. standard Gaussian entries, and consider the random matrix $\Phi = A \circ G$. Then 
$$
\mathbb{P} \left[ \frac{\left| \| \Phi \vct{x} \|_1 - \mathbb{E}_{G}\| \Phi \vct{x} \|_1 \right|}{d} \geq \lambda \right] \leq 2 \exp\left( -\frac{\lambda^2 d}{2 \| \vct{x} \|_2^2} \right), \quad \quad \forall \lambda \geq 0.
$$ 
\end{prop}

\begin{proof}
Recall via Lemma \ref{lemma:mean} that $
\| \Phi x \|_1 = \sum_i | Y_i |$
where $Y_i \sim {\cal N}(0, \sigma_i^2)$ are independent Gaussian random variables with variances $\sigma_j^2 = \sum_{k=1}^n a_{j,k} x_k^2$.  Moreover, $\sum_j \sigma_j^2 = d \| \vct{x} \|_2^2$ thanks to $A \in \{0,1\}^{m \times N}$ having exactly $d$ ones per column. Applying Lemma \ref{concentrate} with $\alpha = \| \vct{x} \|_2$ gives the stated result.
\end{proof}
%

{\subsection{Analytic ingredients}\label{sec:upper}}

\noindent Recall that $G \in \mathbb{R}^{m \times n}$ is a Gaussian random matrix.  We first show that for any binary matrix $A_s \in \{0,1\}^{m \times n}$ having $d = m/s$ ones per column, it holds that $\mathbb{E}_{G} \|(A_s \circ G) \vct{x} \|_1 \leq  1.63 d\sqrt{\frac{2}{\pi}} \| \vct{x} \|_{1,2,s}.$ 

\begin{prop}
\label{exp:upper}
Fix $\vct{x} \in \mathbb{R}^n$ with $s$-block decreasing rearrangement $(\vct{x}_{S_1}, \vct{x}_{S_2}, \dots )$.   Fix $A= (a_{j,k}) \in \{0,1\}^{m \times n}$ having $d = m/s$ ones per column.  Then 
$$
 \sum_{j=1}^m \left( \sum_{k=1}^n a_{j,k} x_{k}^2  \right)^{1/2} \leq  1.63 d \| \vct{x} \|_{1,2,s}.
$$
\end{prop}
For the proof, we will use the following \emph{norm inequality lemma}, which was introduced in \cite{cai2010new}.
\begin{lemma}[From \cite{cai2010new}]
\label{littlelemma}
For any $\vct{x} \in \mathbb{R}^k$,
$$
\| \vct{x} \|_2 \leq \frac{ \| \vct{x} \|_1}{\sqrt{k}} + \frac{\sqrt{k}}{4} \left( \max_{1 \leq i \leq k} |x_i| - \min_{1 \leq i \leq k} |x_i| \right).
$$
\end{lemma}

\begin{proof}[Proof of Proposition \ref{exp:upper}]
We have
\begin{align}
 \sum_{j=1}^m \left( \sum_{k=1}^n a_{j,k} x_{k}^2  \right)^{1/2} &=  \sum_{j=1}^m \left( \sum_{k \in S_1} a_{j,k} x_{k}^2 + \sum_{\ell=2}^{R = \lceil{n/s\rceil}}  \sum_{k \in S_{\ell} }  a_{j,k} x_{k}^2 \right)^{1/2} \nonumber \\
&\leq \sum_{j=1}^m \left( \sum_{k \in S_1} a_{j,k} x_{k}^2 \right)^{1/2} +\sum_{j=1}^m \left( \sum_{\ell=2}^R  \sum_{k \in S_{\ell} }  a_{j,k} x_{k}^2 \right)^{1/2} \nonumber \\
&\leq  \sum_{j=1}^m \sum_{k \in S_1} a_{j,k} |x_{k}| +\sum_{j=1}^m \left( \sum_{\ell=2}^R  \sum_{k \in S_{\ell} }  a_{j,k} x_{k}^2 \right)^{1/2} \nonumber \\
&= d \| \vct{x}_{S_1} \|_1 +\sum_{j=1}^m \left( \sum_{\ell=2}^R  \sum_{k \in S_{\ell} }  a_{j,k} x_{k}^2 \right)^{1/2}  \label{inequalities}
\end{align}
Applying H\"{o}lder's inequality to the second term and using again that $A$ has $d = m/s$ ones per column,
\begin{align}
\label{knorm}
\sum_{j=1}^m \left( \sum_{\ell=2}^R  \sum_{k \in S_{\ell} }  a_{j,k} x_{k}^2 \right)^{1/2} &\leq \sqrt{m} \left(\sum_{j=1}^m \sum_{\ell=2}^R  \sum_{k \in S_{\ell} }  a_{j,k} x_{k}^2 \right)^{1/2} =  d\sqrt{s} \left( \sum_{\ell=2}^R \sum_{k \in S_{\ell}} x_k^2 \right)^{1/2}. 
\end{align}
Applying the norm inequality (Lemma \ref{littlelemma}), we have, for each block $\ell \geq 2$,
\begin{align}
 \sum_{k \in S_{\ell} }  x_{k}^2 &\leq \frac{ \| \vct{x}_{S_{\ell}} \|^2_1}{s} + \frac{s}{16}\| \vct{x}_{S_{\ell}}\|_{\infty}^2 +  \frac{1}{2} \| \vct{x}_{S_{\ell}} \|_1(\| \vct{x}_{S_{\ell}}\|_{\infty} - \| \vct{x}_{S_{\ell+1}} \|_{\infty}) \nonumber \\
  &\leq \frac{ \| \vct{x}_{S_{\ell}} \|^2_1}{s} + \frac{\| \vct{x}_{S_{\ell-1}} \|_1^2}{16s } + \frac{1}{2} \| \vct{x}_{S_{2}} \|_1(\| \vct{x}_{S_{\ell}}\|_{\infty} - \| \vct{x}_{S_{\ell+1}} \|_{\infty}). \nonumber
\end{align}
Returning to the string of inequalities \eqref{inequalities}, we continue
\begin{align}
 \sum_{j=1}^m \left( \sum_{k=1}^n a_{j,k} x_{k}^2  \right)^{1/2} &\leq  d \| \vct{x}_{S_1} \|_1 + \sqrt{m d} \left( \sum_{\ell=2}^R \sum_{k \in S_{\ell}} x_k^2 \right)^{1/2}  \nonumber \\
&\leq d \| \vct{x}_{S_1} \|_1 + \sqrt{m d} \sqrt{   \sum_{\ell=2}^R\left(  \frac{ \| \vct{x}_{S_{\ell}} \|^2_1}{s} + \frac{\| \vct{x}_{S_{\ell-1}} \|_1^2}{16s } +  \| \vct{x}_{S_{2}} \|_1 \left(\frac{\| \vct{x}_{S_{\ell}}\|_{\infty}}{2} -  \frac{\| \vct{x}_{S_{\ell+1}} \|_{\infty}}{2} \right) \right)} \nonumber \\
&\leq d \| \vct{x}_{S_1} \|_1 + d  \sqrt{   \sum_{\ell=2}^R\left( \| \vct{x}_{S_{\ell}} \|^2_1 + \frac{\| \vct{x}_{S_{\ell-1}} \|_1^2}{16} \right) +  \frac{ s \| \vct{x}_{S_{2}} \|_1\| \vct{x}_{S_{2}}\|_{\infty}}{2} } \nonumber \\
&\leq d \| \vct{x}_{S_1} \|_1 +d \sqrt{  \frac{17}{16} \sum_{\ell=2}^R \| \vct{x}_{S_{\ell}} \|^2_1 + \frac{1}{16}  \| \vct{x}_{S_{1}} \|^2_1+ \frac{1}{2}  \| \vct{x}_{S_{1}} \|_1 \| \vct{x}_{S_{2}} \|_1} \nonumber \\
&\leq d \| \vct{x}_{S_1} \|_1 +d \sqrt{  \frac{17}{16} \sum_{\ell=2}^R \| \vct{x}_{S_{\ell}} \|^2_1 + \frac{5}{16}  \| \vct{x}_{S_{1}} \|^2_1+ \frac{1}{4} \|\vct{x}_{S_{2}} \|_1^2} \nonumber \\
&\leq d \| \vct{x}_{S_1} \|_1 +d \sqrt{  \frac{21}{16} \sum_{\ell=2}^R \| \vct{x}_{S_{\ell}} \|^2_1 + \frac{5}{16}  \| \vct{x}_{S_{1}} \|^2_1} \nonumber \\
&\leq \sqrt{2.625}d \sqrt{\sum_{\ell=1}^R \| \vct{x}_{S_{\ell}} \|_1^2}, \label{eq:revtri} \end{align}
where the last inequality follows by applying $|a| + |b| \leq \sqrt{2}\sqrt{a^2 + b^2}$ and the proposition follows by taking the bound $\sqrt{2.625} \leq 1.63$.
\end{proof}

\begin{remark}
\emph{
Note that \eqref{inequalities} and \eqref{knorm} imply that
\begin{align}
\frac{1}{d} \sum_{j=1}^m \left( \sum_{k=1}^n a_{j,k} x_{k}^2  \right)^{1/2} &\leq \| \x_{S_1} \|_1 + \sqrt{s} \| \x_{S_1^c} \|_2 \nonumber \\
&:= {\cal K}(\x, \sqrt{s})
\end{align}
where ${\cal K}(\x, t)$ is the interpolation norm defined in Remark 3.  Given that $\| \vct{x} \|_{1,2,s} \leq {\cal K}(\x ,\sqrt{s})$, a byproduct of the proof of Proposition \ref{exp:upper} is that the two interpolation norms are equivalent up to a factor of 1.63: 
\begin{equation}
\| \vct{x} \|_{1,2,s}  \leq {\cal K}(\x,\sqrt{s}) \leq 1.63 \| \vct{x} \|_{1,2,s} 
\end{equation}
}
\end{remark}

{\subsection{Combinatorial ingredients}\label{sec:balls}}

\noindent Consider $A_s \in \{0,1\}^{m \times n}$ a random binary matrix having $d = m/s$ ones per column.  In this section we show that with high probability with respect to the realization of such a matrix, $\| A_s \vct{x} \|_1 \approx \| \vct{x} \|_1$ for all $s$-sparse $\vct{x} \in \mathbb{R}^n$.   We begin by showing such concentration holds for a fixed $s$-sparse $\vct{x} \in \mathbb{R}^n$.

\bigskip

\begin{prop}
\label{ballsinbins}
Fix $n,s,d \in \mathbb{N}$ and set $m = d s$.   Fix a subset $S \subset [n]$ of cardinality $| S | = s$ and an ordering $\pi_1, \pi_2, \dots, \pi_s$ of the indices in $S$.
Draw an $m \times n$ binary random matrix $A = (a_{j,k})$ with $d = m/s$ ones per column as follows: for each column $k \in [n]$, draw $d$ elements $\{j_1, j_2, \dots, j_d\}$ from $[m]$ uniformly without replacement, and set $a_{j_{\ell}, k} = 1$.  For each $\varepsilon > 0$ it holds with probability exceeding $1 -  2 s \exp(-\varepsilon^2 m/2)$ that,  for any $s$-sparse $\vct{z} \in \mathbb{R}^n$ supported on $S$ and satisfying $|z_{\pi_1}| \geq |z_{\pi_2}| \geq \dots \geq |z_{\pi_s}|$,
\begin{equation}
\label{babyexpand1}
d(1- 2e^{-1} -2\varepsilon) \| \vct{z} \|_1 \leq \| A \vct{z} \|_1 \leq d \| \vct{z} \|_1
\end{equation}
and
\begin{equation*}
d(1 -e^{-1} - \varepsilon) \| \vct{z} \|_1 \leq \sum_{j=1}^m \left( \sum_{k = 1}^n a_{j,k} z_k^2 \right)^{1/2}.
\end{equation*}
\end{prop}

\begin{proof}

As the columns of $A$ are independent random vectors, we can assume without loss that $S = \{1,2,\dots, s\},$ $(\pi_1, \pi_2, \dots, \pi_s) = (1,2, \dots, s)$, and $A = A_S \in \{0,1\}^{m \times s}$. 

\bigskip

For the moment, consider the modified probability distribution over random matrices $A$ where, for each column $k \in [s]$ we draw $d$ elements $\Lambda_k = \{j_1, j_2, \dots, j_d\}$ from $[m]$ uniformly \emph{with} replacement, and set $a_{j_{\ell}, k} = 1$ if $j_{\ell} \in \Lambda_k$ (and $a_{j,\ell} = 0$ otherwise).  Note that by sampling with replacement, there may be repetitions within $\Lambda_k$ and so the number of ones in any particular column may be smaller than $d$.  Still, the total number of draws with replacement is $m = d s$, $d$ draws per each of $s$ columns.   An equivalent way to describe this process is as throwing balls into bins: a total of $m$ balls are thrown i.i.d. into $m$ bins, $d = m/s$ per round for $s$ rounds; if at least one ball is thrown into bin $j$ during round $k$, then $a_{j,k} = 1$; otherwise, $a_{j,k} = 0$. 

\bigskip

\noindent Let $q_k$ be the fraction of the $m$ bins which remain empty after the first $k$ rounds of this process, that is, after the first $d k $ balls have been tossed.  As the probability that any particular bin is empty at this point is equal to $( 1 - \frac{1}{m} )^{d k}$, 
\begin{equation}
\label{q:exp}
\mathbb{E}(q_k) =  \left(1 - \frac{1}{m} \right)^{d k} \leq  \exp\left( - \frac{d k}{m} \right)  = \exp(-k /s).
\end{equation}
 Using the Azuma-Hoeffding inequality, it can be shown [\cite{prob2005},  12.19, p. 313] that the random variable $q_k$ concentrates around its mean according to
 $$
\mathbb{P}\left(  |q_k - \mathbb{E}(q_k) | \geq \varepsilon \sqrt{1 - [\mathbb{E}(q_k) ]^2} \right) \leq 2 \exp(-2 \varepsilon^2 m), \quad \quad \forall \varepsilon > 0.
 $$
Since $\mathbb{E}(q_k) \leq \exp(-k/s)$, this implies in particular that, for any $\varepsilon > 0$,
 \begin{align}
\mathbb{P} \left( q_k \geq \exp(-k/s) + (\varepsilon/\sqrt{2})\sqrt{k/s} \right) \nonumber &\leq  \mathbb{P} \left(  q_k \geq \exp(-k/s) + (\varepsilon/\sqrt{2})\sqrt{1 - \exp(-2k/s)}  \right) \nonumber \\
  &\leq  \mathbb{P} \left(  q_k \geq \mathbb{E}(q_k) + (\varepsilon/\sqrt{2})\sqrt{1 - [\mathbb{E}(q_k)]^2}  \right) \nonumber \\
  &\leq 2 \exp(-\varepsilon^2 m) \quad \quad . \nonumber
  \end{align}
  Taking a union bound over $k \in [s]$,
 \begin{equation}
 \label{union1}
 \mathbb{P} \left( \forall k \in [s]: \hspace{.5mm} q_k   \leq \exp(- k/s) + \varepsilon \sqrt{k / s} \right) \geq 1-  2 s \exp(-\varepsilon^2 m/2), \quad \quad \forall \varepsilon > 0.
 \end{equation}
Back to the setting where we draw $d$ elements $\{j_1, j_2, \dots, j_d\}$ from $[m]$ uniformly \emph{without} replacement to fill each column $k \in [s]$ of $A$, the fraction $\widetilde{q}_k$ of empty rows remaining after $k$ rounds will be even smaller. Specifically, it holds  
$$\mathbb{P} \left( \tilde{q}_k  \leq t \right)  \geq \mathbb{P} \left( q_k \leq t \right), \quad \quad \forall t \geq 0.$$
In turn,
  \begin{equation}
 \label{Noccur}
 \mathbb{P} \left(\forall k \in [s]: \hspace{.5mm} \tilde{q}_k   \leq \exp(- k/s) + \varepsilon \sqrt{k / s} \right) \geq 1 -  2 s \exp(-\varepsilon^2 m/2).
 \end{equation}
 We now assume that the realization of the random matrix $A$ yields, for each $k \in [s]$, $\tilde{q}_k   \leq \exp(- k/s) + \varepsilon \sqrt{k / s}.$  By the above, this occurs with probability exceeding $1 -  2 s \exp(-\varepsilon^2 m/2)$.  Continuing, let $f_k$ be the fraction over $m$ among the first $d k$ balls thrown which form a \emph{collision}, where a ball forms a collision if it lands in a bin which contains a ball thrown from a previous round.  Note that $f_k = \tilde{q}_k- (1-\frac{k}{s})$, the difference between the true fraction of empty bins and the fraction of bins that would be empty if there were no collisions.   It follows that
\begin{align}
f_k \leq \tilde{q}_k - 1 + \frac{k}{s} &\leq \exp(- k /s)  - 1 + k /s + \varepsilon \sqrt{k/s} \nonumber \\
 & \leq \frac{k}{s}e^{-1} +  \varepsilon \sqrt{k/s},
 \end{align}
the last inequality holding because $\exp{(-u)} - 1 + u \leq \exp{(-1)} u$ for $u \in [0,1]$. 

\bigskip

Recall that  $A \in \{0,1\}^{m \times s}$ is such that $a_{j,k} = 1$ if a ball is thrown into bin $j$ during the $k$th round, and $a_{j,k} = 0$ otherwise. For $k\in [s]$, let $E_k \subset [m] \times [k]$ denote the subset of $| E_k | \leq d k$ entries such that $a_{j,\ell} = 1$, and let $C_k \subset E_k$ denote the subset of those entries corresponding to collisions.  Then $|C_k| \leq d k f_k \leq dk (e^{-1}+ \varepsilon)$.  Recall now the assumption that $\vct{z} \in \mathbb{R}^s$ supported on $S = \{1,2,\dots, s\}$ is in decreasing rearrangement: $|z_1| \geq |z_2| \geq \dots \geq |z_s|$.  Observe that we can write
\begin{align}
\label{rhs_sum}
\sum_{(j,k) \in C_s} | z_k | \leq  \sum_{k=1}^s n_k |z_k|
\end{align}
where $n_k \geq 0$ satisfies $\sum_{\ell=1}^k n_{\ell} = | C_k|$.  Given the constraints  $|z_1| \geq |z_2| \geq \dots \geq |z_s|$ and $|C_k | \leq dk (e^{-1}+ \varepsilon)$, the RHS expression of \eqref{rhs_sum} is maximized by setting $n_k = d (e^{-1}+ \varepsilon)$.   Hence,  
\begin{equation}
\label{collisions}
\sum_{(j,k) \in C_s} | z_k | \leq d (e^{-1}+\varepsilon)\| \vct{z} \|_1.
\end{equation}
Since $A$ has $d$ ones per column by construction, 
$$\sum_{(j,k) \in C_s} |z_k| +\sum_{(j,k) \in E_s \setminus C_s} |z_k|   = \sum_{(j,k) \in E_s}  |z_k| = d \| \vct{z} \|_1.
$$
Combined with \eqref{collisions}, 
$$
\sum_{(j,k) \in E_s \setminus C_s} |z_k | \geq d(1- (e^{-1}+\varepsilon)) \| \vct{z} \|_1,
$$
and so 
$$
\| A \vct{z} \|_1 \geq \sum_{(j,k) \in E _s\setminus C_s} |z_k | - \sum_{(j,k) \in C_s} |z_k | \geq  d(1- 2(e^{-1}+\varepsilon)) \| \vct{z} \|_1
$$
and
\begin{equation}
\label{crazy}
\sum_{j=1}^m \left( \sum_{k=1}^s  a_{j,k} |z_k|^2 \right)^{1/2} \geq  \sum_{(j,k) \in E_s \setminus C_s} |z_k|  \geq d(1- (e^{-1}+\varepsilon)) \| \vct{z} \|_1.
\end{equation}
This finishes the proof.
\end{proof}
Applying a union bound over the ${n \choose s} \leq (n/s)^s$ subsets $S \subset [n]$ of size $|S| = s$, and over the $s! \leq s^s$ orderings of indices within any particular such subset $S$, Proposition \ref{ballsinbins} gives rise to a uniform result holding over all $s$-sparse vectors:

\begin{corollary}[Corollary to Proposition \ref{ballsinbins}]
\label{ballsinbins_cor}
Fix $n,s, \in \mathbb{N}$, $\xi \in (0,1)$, and $\varepsilon > 0$.  Fix $d \in \mathbb{N}$ satisfying
$$
d \geq 2\varepsilon^{-2} \log(n/\xi)
$$
and let $m = ds \geq 2s \varepsilon^{-2} \log(n/\xi)$.  Draw an $m \times n$ binary random matrix $A = (a_{j,k})$ with $d$ ones per column chosen uniformly without replacement.  With probability exceeding $1 -  \xi$ it holds that
\begin{equation}
\label{babyexpand}
d(1- 2e^{-1} -2\varepsilon) \| \vct{z} \|_1 \leq \| A \vct{z} \|_1 \leq d \| \vct{z} \|_1 \quad \quad \forall \vct{z} \in \mathbb{R}^n: \| \vct{z} \|_0 \leq s
\end{equation}
and 
\begin{equation*}
d(1 -e^{-1} - \varepsilon) \| \vct{z} \|_1 \leq \sum_{j=1}^m \left( \sum_{k = 1}^n a_{j,k} z_k^2 \right)^{1/2}  \quad \quad \forall \vct{z} \in \mathbb{R}^n: \| \vct{z} \|_0 \leq s
\end{equation*}
\end{corollary}
The inequalities in \eqref{babyexpand} imply that $A/d$ satisfies the 1-restricted isometry property \eqref{1RIP} of order $s$ and level $\theta = 2e^{-1} + 2\varepsilon$. Our probabilistic construction differs from existing constructions of 1-RIP matrices \cite{berinde2008combining, fr13} in that we use $m = sd$ rows but do not seek $\theta$ arbitrarily small, as opposed to using $m \geq \theta^{-1} s d$ rows to achieve 1-RIP for arbitrarily small $\theta > 0$.

\noindent Finally, we show that the results of Proposition \ref{ballsinbins} holding on each of the blocks $\vct{x}_{S_{\ell}}$ in the block-decreasing rearrangement of a vector $\vct{x}$ implies a lower bound on $\mathbb{E}_G \| (A_s \circ G) \vct{x} \|_1$ in terms of  $\| \vct{x} \|_{1,2,s}$.

\begin{prop}
\label{adjacency}
Consider $\vct{x} \in \mathbb{R}^n$ with $s$-block decreasing rearrangement $\vct{x} = ( \vct{x}_{S_1}, \vct{x}_{S_2}, \dots, \vct{x}_{S_{\lceil{n/s\rceil}}} )$. 
Suppose, for some parameter $\gamma \in [0,1]$, that $A \in \mathbb{R}^{m \times n}$ satisfies
\begin{equation}
\label{RIP_like}
d(1 -\gamma) \| \vct{x}_{S_{\ell}} \|_1  \leq \sum_{j=1}^m \left( \sum_{k \in S_{\ell}}^n a_{j,k} x_k^2 \right)^{1/2}, \quad \ell = 1,2, \dots, \lceil{ n/s \rceil}.
\end{equation}
Then 
\begin{equation*}
d (1 - \gamma) \| \vct{x}  \|_{1,2,s}  \leq \sum_{j=1}^m \left( \sum_{k=1}^n a_{j,k} x_{k}^2  \right)^{1/2}.
\end{equation*}
\end{prop}

\begin{proof}    
Consider the vector $\vct{y}_j = (y_{j,\ell}) \in \mathbb{R}^{\lceil{n/s\rceil}}$ with coordinates $y_{j,\ell} = \left( \sum_{k \in S_{\ell}} a_{j,k} x_k^2 \right)^{1/2}.$ By the triangle inequality, 
$$
\| \sum_{j=1}^m \vct{y}_j \|_2 = \sqrt{\sum_{\ell=1}^{\lceil{n/s \rceil}} \left( \sum_{j=1}^m \left( \sum_{k \in S_{\ell}} a_{j,k} x_k^2 \right)^{1/2} \right)^2} \leq  \sum_{j=1}^m \left( \sum_{k=1}^n a_{j,k} x_k^2 \right)^{1/2}  = \sum_{j=1}^m \| \vct{y}_j \|_2 
$$
Incorporating the assumptions \eqref{RIP_like} and recalling that $\| \vct{x} \|_{1,2,s} = \left( \sum_{\ell=1}^{ \lceil{n/s\rceil}} \| \vct{x}_{S_{\ell}} \|_1^2 \right)^{1/2}$ gives the desired result.
\end{proof}

{\section{Proof of Theorem \ref{main}} \label{proof:thm}}
\noindent In this section we put together the ingredients to prove Theorem \ref{main}.
Fix $\vct{x} \in \mathbb{R}^n$ with $s$-block decreasing rearrangement $\vct{x} = (\vct{x}_{S_1}, \vct{x}_{S_2}, \dots, \vct{x}_{S_{\lceil{n/s\rceil}}})$.   Fix parameter $\varepsilon > 0$. 
Consider $A = (a_{j,k}) \in \{0,1\}^{m \times n}$ a random binary matrix having $d$ ones per column as constructed in Proposition \ref{ballsinbins}.  From that proposition, 
\begin{align}
\label{firstprob}
\mathbb{P} & \left( \exists \ell \in \{1,2, \dots, \lceil{n/s \rceil}: \hspace{2mm} d(1 -e^{-1} - \varepsilon/2) \| \vct{x}_{S_{\ell}} \|_1 > \sum_{j=1}^m \left( \sum_{k \in S_{\ell}} a_{j,k} x_k^2 \right)^{1/2} \right) \nonumber \\
&\leq \sum_{\ell=1}^{\lceil{n/s \rceil}} \mathbb{P} \left( d(1 -e^{-1} - \varepsilon/2) \| \vct{x}_{S_{\ell}} \|_1 > \sum_{j=1}^m \left( \sum_{k \in S_{\ell}} a_{j,k} x_k^2 \right)^{1/2} \right) \nonumber \\
&\leq 2(n/s+1)s \exp( - \varepsilon^2 m / 8 ) \nonumber \\
&\leq 4n \exp( - \varepsilon^2 m / 8 ).
\end{align}
We now assume that the realization of the random matrix $A$ yields
$$
d(1 -e^{-1} - \varepsilon/2) \| \vct{x}_{S_{\ell}} \|_1 \leq \sum_{j=1}^m \left( \sum_{k \in S_{\ell}} a_{j,k} x_k^2 \right)^{1/2}, \quad \ell = 1,2, \dots, \lceil{n/s\rceil},
$$
and hence, by Proposition \ref{adjacency},
$$
d(1 -e^{-1} - \varepsilon/2) \| \vct{x} \|_{1,2,s} \leq \sum_{j=1}^m \left( \sum_{k =1}^n a_{j,k} x_k^2 \right)^{1/2}.
$$
By the above, this occurs with probability exceeding $1 - 4n \exp( - \varepsilon^2 m / 8 )$.

\bigskip

Consider now $G \in \mathbb{R}^{m \times n}$ having i.i.d. standard Gaussian entries, and let $\Phi = A \circ G$ where $\circ$ denotes the Hadamard (entrywise) product.  By Lemma \ref{lemma:mean}, $ \mathbb{E}_{G} \| \Phi \vct{x} \|_1 = \beta_0 \sum_{j=1}^m \left( \sum_{k =1}^n a_{j,k} x_k^2 \right)^{1/2}$ where $\beta_0 = \sqrt{2/\pi}$, and so by the above analysis it follows
$$\beta_0 d (1 - e^{-1} - \varepsilon/2) \| \vct{x}  \|_{1,2,s}  \leq  \mathbb{E}_{G} \| \Phi \vct{x} \|_1.$$
\noindent By Proposition \ref{exp:upper},  we have also the upper bound
\begin{equation}
\label{holdingbelow}
\mathbb{E}_{G} \| \Phi \vct{x} \|_1 \leq 1.63 \beta_0 d \| \vct{x} \|_{1,2,s}.
\end{equation}
Now, Proposition \ref{main:concentrate} gives that with respect to the draw of $G$, 
\begin{align}
\label{prob2}
\mathbb{P} \left( \left|  \| \Phi \vct{x} \|_1 - \mathbb{E}_{G} \| \Phi \vct{x} \|_1 
  \right| \geq \frac{\varepsilon}{2} \beta_0 d \| \vct{x} \|_{1,2,s}  \right) \leq 2 \exp\left( \frac{-\varepsilon^2 \beta_0^2 \| \vct{x} \|_{1,2,s}^2 d}{8 \| \vct{x} \|_2^2} \right).
  \end{align}
  We now assume that the realization of the random matrix $G$ yields
  $$
   \left|  \| \Phi \vct{x} \|_1 - \mathbb{E}_{G} \| \Phi \vct{x} \|_1 
  \right| \leq \frac{\varepsilon}{2} \beta_0 d \| \vct{x} \|_{1,2,s},
  $$
  which by the above occurs with probability exceeding $1-2 \exp\left( \frac{-\varepsilon^2 \beta_0^2 \| \vct{x} \|_{1,2,s}^2 d}{8 \| \vct{x} \|_2^2} \right).$
Adding together the probabilities that either our assumption on $A$ or our assumption on $G$ does not hold, we have shown that with probability exceeding $1 - 4n \exp( - \varepsilon^2 m / 8 ) -  2 \exp\left( \frac{-\varepsilon^2 \beta_0^2 m}{8s}(\frac{\| \vct{x} \|_{1,2,s}}{ \| \vct{x} \|_2})^2 \right),$
$$
d \beta_0 (1 - e^{-1} - \varepsilon) \| \vct{x} \|_{1,2,s}\leq \| \Phi \vct{x} \|_1 \leq d \beta_0 (1.63+ \varepsilon) \| \vct{x} \|_{1,2,s}.
$$
Setting $\Psi_s =  \frac{1}{d \beta_0} (A \circ G),$ recalling that $d = m/s,$ and using the bound $.63 \leq 1 - e^{-1},$ we recover the content of Theorem \ref{main}.\\

\section{Discussion}
Theorem~\ref{main} introduces a family of maps, which are shown to map finite-dimensional spaces equipped with the block $\ell_1/\ell_2$-norm to lower dimensional spaces equipped with the the $\ell_1$-norm, while with high probability preserving the norm up to a constant distortion factor. In Euclidean space, by contrast, Johnson-Lindenstrauss embeddings can be made to have distortion arbitrarily close to 1, that is, there is only a factor that can be made arbitrarily small by increasing the embedding dimension.

In this section we show that some distortion factor is indeed necessary; we give two explicit families of examples in arbitrarily large dimensions for which the fraction of their block $\ell_1/\ell_2$-norm and the $\ell_1$-norm of their image behave differently even in the asymptotic limit. 
Conditioned on the (dependent) random variables $a_{j,k}$, the entries of $\Phi_s$ and hence also the entries of $\Phi_s \vct{x}$ are independent Gaussian random variables. Thus the concentration of $\|\Phi_s \vct{x}\|_1$ around its mean is precisely understood, and it remains to compare the behavior of the mean for different instances of $\vct{x}$. As shown in Lemma~\ref{lemma:mean} above, this boils down to studying the quantity $\sum_{j=1}^m \left( \sum_{k=1}^n a_{j,k} x_{k}^2  \right)^{1/2}.$ The following proposition indeed provides two vectors with significantly different behavior of this quantity as normalized by the block $\ell_1/\ell_2$-norm.

\begin{prop} \label{prop:counterex} Choose $a_{j,k}$ as defined in Theorem~\ref{main} and assume that $m\leq n$. Then for each $\eta,\nu >0$, there exists a constant $C>0$ such that if $s>C$ and $\tfrac{n}{s \log(n)}>C$, the following holds with probability at least $1-\nu$.

Consider $\vct{x}, \vct{y}\in \R^{n+1}$ with 
 \[
  \vct{x} = \big(1, \frac{1}{\sqrt{ns}}, \frac{1}{\sqrt{ns}}, \dots, \frac{1}{\sqrt{ns}}\big)^t \text{ and } \vct{y}=(1, 0, \dots, 0)^t.
 \]
 Then one has $\tfrac{\sum_{j=1}^m \left( \sum_{k=1}^n a_{j,k} x_{k}^2  \right)^{1/2}}{d\|\vct{x}\|_{1,2,s}}\geq \sqrt{2}-\eta$, while $\tfrac{\sum_{j=1}^m \left( \sum_{k=1}^n a_{j,k} y_{k}^2  \right)^{1/2}}{d\|\vct{y}\|_{1,2,s}}\equiv 1$.
\end{prop}

\begin{proof}
To facilitate the calculations, we assume that $s$ divides $n+1$, if this is not the case, one obtains an only slightly changed result. We first determine the block $\ell_1/\ell_2$-norm of $\vct{x}$ and $\vct{y}$, obtaining $\|\vct{y}\|^2_{1,2,s}=1$ and
\begin{equation}
 \|\vct{x}\|^2_{1,2,s}= (1+\frac{s-1}{\sqrt{s n}})^2 + (\frac{n+1}{s}-1)\frac{s}{n} = 2 +2\frac{s-1}{\sqrt{s n}} -\frac{1}{n} +\frac{1}{ns}\leq 2+2\sqrt{\frac{s}{n}}.
\end{equation}
To estimate the numerators, we note that for $j$ fixed, the $a_{j,k}$ are independent Bernoulli random variables with parameter $p=\tfrac{1}{s}$. Thus
$\E \sum_{k=2}^{n+1} a_{j,k} x_{k}^2 =\tfrac{1}{s^2}$ and, for $\delta>0$,
\[
 \P\left( \sum_{k=2}^{n+1} a_{j,k} x_{k}^2  <(1- \delta) \frac{1}{s^2} \right) = \P\left( \sum_{k=2}^{n+1} a_{j,k} <(1- \delta)\frac{n}{s}\right).
\]
The latter is a large deviation probability for the binomial distribution. It can be bounded via the relative entropy between two biased coins
\begin{equation}
 H(a,p)=a \log(\frac{a}{p})+(1-a)\log(\frac{1-a}{1-p}).
\end{equation}
Then (for example by Theorem 1 in \cite{AG}, applied to the expression with the roles of ones and zeros exchanged), the large deviation probability is bounded by
\begin{equation}
 \P\left( \sum_{k=2}^{n+1} a_{j,k} < (1-\delta)\frac{n}{s} \right) \leq \exp\left(-n H\left( \frac{1+\delta}{s}, \frac{1}{s} \right)\right) \leq \left( \frac{1}{1+\delta} \right)^{(1+\delta)n/s} =e^{-cn/s},
\end{equation}
where $c=(1+\delta)\log(1+\delta)>0$.
Now we know that $a_{j,1}=1$ for exactly $d$ randomly chosen values of $j$. For these values of $j$, one has
\[
 \P\left( |\sum_{k=1}^{n+1} a_{j,k} x_{k}^2  < 1+(1- \delta) \tfrac{1}{s^2} \right) \leq e^{-cn/s}, 
\]
while for the other values of $j$, one has
\[
 \P\left( \sum_{k=1}^{n+1} a_{j,k} x_{k}^2   < (1-\delta)\tfrac{1}{s^2}  \right) \leq e^{-cn/s}.
\]
Thus with probability at least $1-2m e^{-cn/s}$, one has
\begin{align*}
\sum_{j=1}^m \left( \sum_{k=1}^{n+1}a_{j,k} x_{k}^2  \right)^{1/2} &\geq d \sqrt{1+(1-\delta)\tfrac{1}{s^2}} + (m-d)\sqrt{1-\delta}\tfrac{1}{s}\\&\geq
d (1+\sqrt{1-\delta}- \tfrac{1}{s}(\sqrt{1-\delta}))\\
&\geq d(2-\delta -\tfrac{1}{s})
\end{align*}
and consequently 
\begin{align*}
 \tfrac{\sum_{j=1}^m \left( \sum_{k=1}^n a_{j,k} x_{k}^2  \right)^{1/2}}{d \|\vct{x}\|_{1,2,s}} & \geq \frac{2-\delta -\tfrac{1}{s}}{\sqrt{2+2\sqrt{s/n}}}.
\end{align*}
It is clear that if $s$ and $\tfrac{n}{s}$ are large enough and if $\delta$ is chosen small enough, then this expression is ensured to get arbitrarily close to $\sqrt{2}$, as desired. On the other hand, as $m\leq n$, the associated probability of failure is bounded by $m e^{-cn/s}\leq e^{\log(n) - cn/s}$, which, for fixed $c$ and $s$, becomes arbitrarily small for $n$ large enough. 

The estimate for $\vct{y}$ follows directly from the fact that $a_{j,1}=1$ for exactly $d$ values of $j$ and $0$ otherwise.
\end{proof}

   \begin{remark}
    The definition of $\vct{x}$ is inspired by the proof of Proposition~\ref{exp:upper}. Namely, the first inequality in \eqref{inequalities} is sharp when the two terms are equal, and the last inequality in \eqref{eq:revtri} is the sharper, the more the two summands differ.  So the constant in Proposition~\ref{exp:upper} can be significantly improved unless both of these facts happen at the same time. This basically boils down to having a jump within the first block and roughly the same block norm contribution of the first block and the tail, which $\vct{x}$ is an extreme example of.
   \end{remark}

Proposition~\ref{prop:counterex} provides a counterexample, for which a constant distortion is necessary in Theorem~\ref{main}. Considering that our proofs above yield the exact analogue to Theorem~\ref{main} also for the interpolation norm $K(\vct{x},\sqrt{s})$ instead of the block norm $\|\vct{x}\|_{1,2,s}$, one may ask whether the former norm has better empirical performance. To answer this question, we numerically test the validity of Theorem~\ref{main} for these two norms on four different types of signals for $n=1000$, $s=10$ fixed and varying number of measurements. The results are presented in Figure~\ref{fig:norms}; for each of the norms, we plot the relative distortion $\tfrac{\|\Psi_s \vct{x}\|_1 - \| \vct{x} \|}{\| \vct{x} \|}$, where $\Psi_s$ is as in Theorem~\ref{main} and $\|\cdot\|$ is the norm in question.

In example (a),  $s$-sparse signals are generated by choosing a support at random and the corresponding entries according to the standard normal distribution. In example (b), we consider the fixed $1$-sparse vector supported in the first position. In example (c), we consider signals inspired by the counterexample of Proposition~\ref{prop:counterex} with one large entry and the other entries chosen according to a normal distribution with considerably smaller variance. Finally, in example (d), we consider signal with independent entries drawn from the standard normal distribution. Our experiments show that for both norms, the average $\ell_1$-norm of the image of normalized vectors in the different classes behave differently. As expected, for the block norm, the signals inspired by the counterexample yield images with a larger $\ell_1$-norm than for $1$-sparse and also for random signals. On the other hand, for the interpolation norm, images of random signals typically have smaller norm, while the vectors 
inspired by $\vct{x}$ and $\vct{y}$ in Proposition~\ref{prop:counterex} have a comparable behavior. These observations made us choose to present our results in terms of the block norm rather than the interpolation norm, as for the block norm, the upper and lower distortion factors essentially disappear for random signals (which we see as representing the generic behavior). Notably, in the example of $s$-sparse signals, the resulting images have a smaller norm. In this case, the behavior for the two norms is identical, as for sparse vectors they both reduce to the $\ell_1$-norm.

\begin{figure}[h!]
\includegraphics[width=14cm]{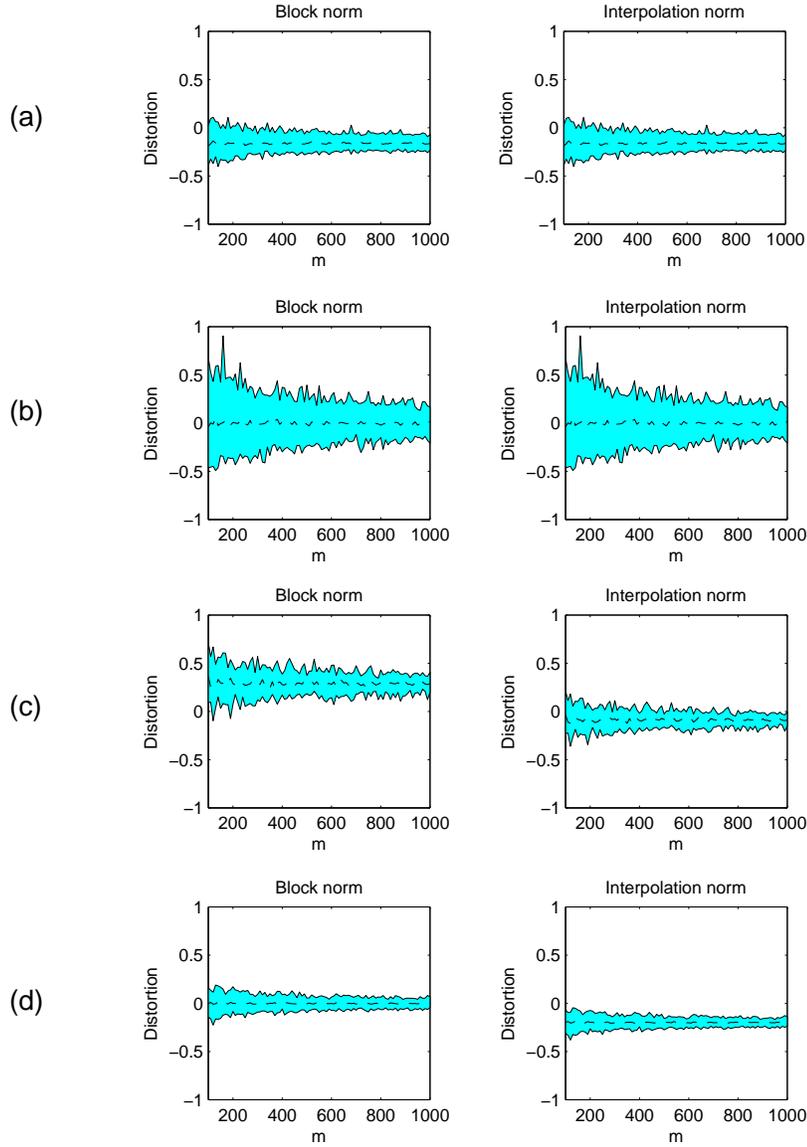}
\caption{A comparison of the distortion $(\| \Psi_s(x) \|_1 -\| x \|)  / \| x \|$ where $\| x \|$ is either the block norm $\| x \|_{1,2,s}$ or the interpolation norm ${\cal K}(x,\sqrt{s})$ with $\Psi_s$ as in Theorem \ref{main} and $s = 10$ fixed throughout. We consider distortion statistics (the minimum, maximum, and median) over $64$ instances of $\Psi_s$ over several signal classes: (a) $s$-sparse signals, (b) 1-sparse signals, (c ) a two-level signal inspired by Proposition \ref{prop:counterex}, and (d) Gaussian random signals.}
\label{fig:norms}
\end{figure}

\section*{Acknowledgment}
We thank Anna Gilbert, Arie Israel, Joe Neeman, Jelani Nelson, and Mark Rudelson for helpful input and suggestions.  We also thank the anonymous referees for their valuable comments.  F.~Krahmer was supported by the German Science Foundation (DFG) in the  context of the Emmy Noether Junior Research Group ``RaSenQuaSI'' (grant KR 4512/1-1).  R.~Ward was supported in part by an AFOSR Young Investigator Award, DOD-Navy grant N00014-12-1-0743, and an NSF CAREER Award.

\bibliography{biblio}
\bibliographystyle{abbrv}

\section{Appendix}

\subsection{Verifying that $\| \cdot \|_{1,2,s}$ is a norm}
Here we verify that the function $\| \vct{x} \|_{1,2,s} := \sqrt{ \sum_{\ell =1}^{\lceil{n/s \rceil}} \| \vct{x}_{S_{\ell}} \|_1^2 }$ is indeed a norm.

\bigskip

It is straightforward that $\| \vct{x} \|_{1,2,s} \geq 0$ and that $\| \vct{x} \|_{1,2,s} = 0$ implies $\vct{x} = 0$.  It is also clear that $\|  a \vct{x} \|_{1,2,s} = | a | \| \vct{x} \|_{1,2,s}$.  It remains to verify the triangle inequality: $\| \vct{x} + \vct{y} \|_{1,2,s} \leq \| \vct{x} \|_{1,2,s} + \| \vct{y} \|_{1,2,s}$ for any $\vct{x}, \vct{y} \in \mathbb{R}^n$.  
To do this, let us set up some notation. Denote the partitioned supports in the block decreasing rearrangement of $\vct{x}$ by $S_1, S_2, \dots $, the partitioned supports in the block decreasing rearrangement of $\vct{y}$ by $T_1, T_2, \dots $ and the partitioned supports of the block decreasing rearrangement of $\vct{x+y}$ by $U_1, U_2, \dots$.   
 Then
\begin{align}
\| \vct{x + y} \|_{1,2,s} &= \sqrt{ \sum_{\ell=1}^{\lceil{ n/s \rceil}} \| (\vct{x+y})_{U_{\ell}} \|_1^2 }  \leq \sqrt{ \sum_{\ell=1}^{\lceil{ n/s \rceil}} \| \vct{x}_{U_{\ell}} \|_1^2 } + \sqrt{ \sum_{\ell=1}^{\lceil{ n/s \rceil}} \| \vct{y}_{U_{\ell}} \|_1^2 } \nonumber 
\end{align}
thanks to the block $\ell_1/\ell_2$ vector norm (with support sets fixed) satisfying the triangle inequality.   Now, $\sum_{\ell=1}^{\lceil{ n/s \rceil}} \| \vct{x}_{U_{\ell}} \|_1^2 \leq  \sum_{\ell=1}^{\lceil{ n/s \rceil}} \| \vct{x}_{S_{\ell}} \|_1^2 $ can be seen to hold by appealing to \emph{Karamata's inequality} \cite{k1} to the sequences $\vct{r}_1 = (\| \vct{x}_{S_1} \|_1, \| \vct{x}_{S_2} \|_1, \dots )$ and $\vct{r}_2 = ( \| \vct{x}_{U_1} \|_1, \| \vct{x}_{U_2} \|_1, \dots)$, noting that $\vct{r}_1$ majorizes $\vct{r}_2$.  Using the same argument to show $\sum_{\ell=1}^{\lceil{ n/s \rceil}} \| \vct{y}_{U_{\ell}} \|_1^2 \leq  \sum_{\ell=1}^{\lceil{ n/s \rceil}} \| \vct{y}_{T_{\ell}} \|_1^2 $,  we then have that the RHS expression above is 
\begin{align}
&\leq  \sqrt{ \sum_{\ell=1}^{\lceil{ n/s \rceil}} \| \vct{x}_{S_{\ell}} \|_1^2 } + \sqrt{ \sum_{\ell=1}^{\lceil{ n/s \rceil}} \| \vct{y}_{T_{\ell}} \|_1^2 } = \| \vct{x} \|_{1,2,s} + \| \vct{y} \|_{1,2,s},
\end{align}
verifying the triangle inequality.
 
 \subsection{Proof of Lemma \ref{concentrate}}
Recall that $Y_i \sim {\cal N}(0, \sigma_i^2)$ has density function is $dF_i(t) =\beta_0 \frac{1}{2\sigma_i} \exp{(-\frac{t^2}{2\sigma_i^2})} dt$ where   $\beta_0 = \sqrt{2/\pi}$.  We may then estimate for each $u \geq 0$
\begin{align}
\mathbb{E} \left[ \exp{(u (|Y_i| - \beta_0 \sigma_i)} \right]  &= \int_{t = -\infty}^{\infty} \exp\left( u(|t|-\beta_0 \sigma_i) \right) dF_i(t) \nonumber \\
&= \frac{\beta_0}{\sigma_i} \int_{t=0}^{\infty} \exp(u(t-\beta_0 \sigma_i))\exp{\left(-\frac{t^2}{2\sigma_i^2}\right)} dt \nonumber \\
(\text{set } s = \sqrt{2} t / \sigma_i - \sqrt{2} \sigma_i u ) \quad \quad \quad &=   \exp{(u^2 \sigma_i^2/2 - u \beta_0 \sigma_i)} \sqrt{2} \left[ (\beta_0/2)\int_{s=- \sqrt{2} \sigma_i u}^{\infty} \exp{(-s^2/4)} ds \right] \nonumber \\
&=   \exp{(u^2 \sigma_i^2/2 - u \beta_0 \sigma_i)} \sqrt{2} \left[1/2 + (\beta_0/2)\int_{s=0}^{ \sqrt{2}\sigma_i u} \exp{(-s^2/4)} ds \right] \nonumber
\end{align}
One of the two cases holds:
\begin{enumerate}
\item If $u \geq \frac{\log(2)}{2 \beta_0 \sigma_i}$, then $\sqrt{2} \leq \exp(u \beta_0 \sigma_i)$
\item If $0 \leq u <  \frac{\log(2)}{2 \beta_0 \sigma_i}$, then 
$$\sqrt{2} \left[1/2 + (\beta_0/2)\int_{s=0}^{ \sqrt{2} \sigma_i u} \exp{(-s^2/4)} ds \right] \leq 1/\sqrt{2} + \frac{\beta_0}{\sqrt{2}} \sqrt{2}\sigma_i u \leq 1 $$
\end{enumerate}
In either case, we may bound the final RHS expression above to estimate 
\begin{align}
\mathbb{E} \left[ \exp{(u (|Y_i| - \beta_0 \sigma_i)} \right]  &\leq   \exp{(u^2 \sigma_i^2/2)}.
\end{align}
A similar analysis reveals that also
\begin{align}
\mathbb{E} \left[ \exp{(-u (|Y_i| - \beta_0 \sigma_i)} \right]  &\leq  \exp{(u^2 \sigma_i^2/2)}.
\end{align}

\noindent Recall that the random variable of interest is of the form $Z = \frac{1}{d} \sum_i ( |Y_i| - \mathbb{E} |Y_i| ) =  \frac{1}{d} \sum_i ( |Y_i| - \beta_0 \sigma_i )$ with $Y_i \sim {\cal N}(0, \sigma_i^2)$ independent.  Recall also that $\sum_i \sigma_i^2 \leq d \alpha^2$ by assumption. It follows that
\begin{align}
\mathbb{E} (\exp{(u Z)} ) &=  \mathbb{E} \exp\left( {\frac{u}{d} \sum_i ( |Y_i| - \beta_0 \sigma_i)} \right) \nonumber \\
&=  \mathbb{E} \prod_{i=1}^m \exp\left( {\frac{u}{d} ( |Y_i| - \beta_0 \sigma_i)} \right) \nonumber \\
&= \prod_{i=1}^m \mathbb{E} \left[ \exp\left( {\frac{u}{d} ( |Y_i| - \beta_0 \sigma_i)} \right) \right] \nonumber \\
&\leq  \prod_{i=1}^m  \exp{\left(\frac{u^2 \sigma_i^2}{2d^2} \right)} \nonumber \\
&= \exp\left(\frac{u^2}{2d^2} \sum_i \sigma_i^2 \right) \leq \exp\left(\frac{u^2  \alpha^2}{2d} \right)\nonumber 
\end{align}
Recall that if a random variable $X$ satisfies $\mathbb{E}[X] = 0$ and $\mathbb{E} \exp(u X) \leq \exp(C u^2)$ for all $u \in \mathbb{R}$ and for some constant $C > 0$, then $\mathbb{P}( |X| \geq \lambda) \leq 2 \exp(-\frac{\lambda^2}{4C})$ for each $\lambda \geq 0$ (see, for example, Proposition 7.24 of \cite{fr13}).  Since $\mathbb{E}[Z] = 0$, it follows that $\mathbb{P}( |Z| \geq \lambda) \leq 2 \exp(-\frac{\lambda^2 d}{2 \alpha^2})$ for each $\lambda \geq 0$. This proves Lemma \ref{concentrate}.

\end{document}